\documentclass[letterpaper,onecolumn,journal,draft]{IEEEtran}

\usepackage{cite}
\usepackage{graphicx}

\usepackage{amsmath}
\usepackage{amsfonts}
\usepackage{amssymb}
\usepackage{bm}

\usepackage{amsthm}
\theoremstyle{plain}
\newtheorem{thm}{Theorem}[section]

\newtheorem*{cor}{Corollary}
\theoremstyle{definition}
\newtheorem{defn}{Definition}[section]

\newtheorem{exmp}{Example}[section]
\theoremstyle{remark}
\newtheorem*{rem}{Remark}

\title{A Cram\'{e}r-Rao Bound for Semi-Blind \\Channel Estimation in \\Redundant Block Transmission Systems}
\author{Yen-Huan Li, \IEEEmembership{Member, IEEE, }Borching Su, \IEEEmembership{Member, IEEE, }and Ping-Cheng Yeh, \IEEEmembership{Member, IEEE}
}

\begin{document}
\maketitle
\begin{abstract}
A Cram\'{e}r-Rao bound (CRB) for semi-blind channel estimators in redundant block transmission systems is derived. The derived CRB is valid for any system adopting a full-rank linear redundant precoder, including the popular cyclic-prefixed orthogonal frequency-division multiplexing system. Simple forms of CRBs for multiple complex parameters, either unconstrained or constrained by a holomorphic function, are also derived, which facilitate the CRB derivation of the problem of interest. The derived CRB is a lower bound on the variance of any unbiased semi-blind channel estimator, and can serve as a tractable performance metric for system design.
\end{abstract}

\section{Introduction}

We derive a variance lower bound for unbiased semi-blind channel estimators in redundant block transmission systems, practical designs of which can be found in \cite{Zhou2001,Muquet2002,Chang2004,Zeng2004,Su2007b,Shin2008,Wan2008,Sarmadi2009,Yu2009,Kim2009,Wan2011}. We formulate the semi-blind channel estimation problem as a non-Bayesian parametric estimation problem with constrained parameter space, and derive the constrained Cram\'{e}r-Rao bound (CRB) for unbiased channel estimators explicitly. The derived bound is valid for any full-rank linear redundant precoder, including the popular cyclic-prefixed orthogonal frequency-division multiplexing (CP-OFDM) system adopted in IEEE 802.11n \cite{802.11n}, IEEE 802.16e \cite{IEEE802.16e} and 3GPP-LTE \cite{3GPPLTE}. 

The derived CRB can be used as a tractable performance metric in communication system design as in \cite{Dong2002,Stoica2003}, where pilot symbols are designed with respect to the resulted CRBs. Note that the problem settings in \cite{Dong2002,Stoica2003} are different from ours. In \cite{Stoica2003} the authors consider the non-blind approach instead of the semi-blind one. The authors of \cite{Dong2002} also focus on the semi-blind approach, but their definition of pilot symbols is not as general as ours (cf. (\ref{eq_def_pilot_symbol})), so they do not need to derive a constrained CRB for complex parameters (see Theorem \ref{thm_constrained_crb} in this paper). Furthermore, they adopt the Bayesian approach and model the channel as a random vector of certain probability distribution, while in this paper we adopt a non-Bayesian approach, and the information of channel distribution is not required.

A CRB is usually viewed as a lower bound of the optimal mean square error (MSE) performance, but this is not the case if we do not specify the bias function of estimators. Most of the existing channel estimators, such as those proposed in \cite{Zhou2001,Muquet2002,Chang2004,Zeng2004,Su2007b,Shin2008,Wan2008,Sarmadi2009,Yu2009,Kim2009,Wan2011}, are \emph{possibly biased with unknown bias functions}. Since MSE is the sum of the variance and the square of the bias of an estimator, an appropriately designed biased estimator may outperform all unbiased estimators, and any CRB for unbiased estimators, in terms of MSE \cite{Stoica1990,Eldar2008,Todros2011}. This is why we focus more on the role of CRB as a tractable performance metric rather than a performance lower bound in terms of MSE.


To calculate the CRB for unbiased semi-blind channel estimators in a wireless communication system, which is usually modeled by complex numbers in a discrete-time baseband representation, we need to extend the existing results of CRBs for unconstrained and constrained parameters, originally for real parameters \cite{Kay1993,VanTrees2001,Gorman1990,Marzetta1993,Stoica1998}, to the case of multiple complex parameters. There has been some literature focusing on this topic \cite{Bos1994,Jagannatham2004,Ollila2008,Schreier2010,Smith2005}. In \cite{Bos1994,Jagannatham2004,Ollila2008,Schreier2010}, the derived CRBs for unconstrained and constrained complex parameters, unlike the corresponding CRBs for real parameters, are variance lower bounds for any unbiased estimator for a $(2n)$-dimensional complex parameter $[\bm{\theta}^T \, \bm{\theta}^H]^T$, instead of the original $n$-dimensional complex parameter $\bm{\theta}$ of interest. Computational complexity of calculating the CRB is increased, and additional calculation is required to extract the CRB of any unbiased estimator of $\bm{\theta}$, which hinders us from obtaining a closed-form or insights when deriving the bounds. In \cite{Smith2005}, a variance lower bound is derived directly for unbiased estimators of $\bm{\theta}$ for the first time, but the result is in a complicated form compared with the well-known CRB for real parameters \cite{Kay1993,VanTrees2001}.

An additional contribution of this paper is a simplification of CRB formulae for multiple complex parameters. Two cases are considered here. In the first case the parameter space is unconstrained; that is, the parameter space is the set of all $n$-dimensional complex vectors. In the second case the parameter space is constrained by some \emph{holomorphic} function. The derived results are \emph{simple} for two reasons. First, they share exactly the same forms with their real counterparts, and thus avoid the use of a double-sized parameter. Second, since the derivation procedure is valid for both real and complex parameters, the simple forms enables us to extend existing CRBs in various situations, originally for real parameters, to the case of complex parameters without modifying the forms.

The rest of the paper is organized as follows. In Section \ref{sec_complex_crb}, we derive simple CRBs for unconstrained and constrained complex parameters. In Section \ref{sec_syst_model_problem_formulation}, we describe the discrete-time baseband system model of a redundant block transmission system, and formulate the semi-blind channel estimation problem as a non-Bayesian parametric estimation problem with a constrained parameter space. We then derive the CRB for semi-blind channel estimators in Section \ref{sec_crb_for_chEst}, applying results obtained in Section \ref{sec_complex_crb}. Conclusions are presented in Section \ref{sec_conclusion}.

\subsection*{Notation} \label{sec_notation}
Bold-faced lower case letters represent column vectors, and bold-faced upper case letters are matrices. Superscripts such as in $\bm{v}^*$, $\bm{v}^T$, $\bm{v}^H$, $\bm{M}^{-1}$, and $\bm{M}^{\dagger}$ denote the conjugate, transpose, conjugate transpose (Hermitian), inverse, and Moore-Penrose generalized inverse of the corresponding vector or matrix. The Kronecker product of $\bm{A}$ and $\bm{B}$ is denoted by $\bm{A} \otimes \bm{B}$. Matrices $\bm{I}_n$ and $\bm{0}_{m\times n}$ denote the $n\times n$ identity matrix and the $m\times n$ zero matrix, respectively. Notation $\mathcal{T}_{\bm{v}}$ denotes a Toeplitz matrix 
\begin{equation}
\left[
\begin{array}{cccc}
v_0 & 0 & \cdots & 0 \\
v_1 & \ddots & \ddots & \vdots \\
\vdots & \ddots & \ddots & 0 \\
v_L & \ddots & \ddots & v_0 \\
0 & \ddots & \ddots & v_1 \\
\vdots & \ddots & \ddots & \vdots \\
0 & \ldots & 0 & v_L
\end{array}
\right], \notag
\end{equation}
where $v_i$ is the $i$th element of the $L$-dimensional vector $\bm{v}$. The notation $\bm{A} \geq \bm{B}$ means that $\bm{A} - \bm{B}$ is a nonnegative-definite matrix. The vector $\mathsf{E}\left[ \bm{v} \right]$ denotes the expectation of $\bm{v}$, and the matrix $\mathsf{E}\left[ \bm{M} \right]$ denotes the expectation of $\bm{M}$. The matrix $\mathsf{cov}(\bm{u}, \bm{v})$ is defined as $\mathsf{E}[ ( \bm{u} - \mathsf{E}[\bm{u}] ) ( \bm{v} - \mathsf{E}[\bm{v}] )^H ]$, and $\mathsf{cov}( \bm{v}, \bm{v} )$ is denoted by $\mathsf{cov}( \bm{v} )$ for brevity.

\section{Simple Forms of CRBs for Complex Parameters} \label{sec_complex_crb}
\subsection{Introduction}
Let $\bm{y}$ be a sample from some probability density function (pdf) $p( \bm{y}; \bm{\theta} )$, which belongs to a family of pdf's $\left\{ p( \bm{y}; \bm{\theta} ), \bm{\theta} \in \Theta \right\}$, parametrized by a parameter space $\Theta$. The CRB shows that if both $\bm{y}$ and $\bm{\theta}$ are \emph{real} vectors, the variance of any unbiased estimator $\hat{\bm{\theta}}( \bm{y} )$ of $\bm{\theta}$ must follow the following inequality.
\begin{equation}
\mathsf{cov}( \hat{\bm{\theta}} ) \geq \bm{J}^{-1} := \left( \mathsf{E} \left[ \bm{v} \bm{v}^T \right] \right)^{-1}, \notag
\end{equation}
where the random vector $\bm{v}( \bm{y}; \bm{\theta} ) := {\partial \ln p(\bm{y}; \bm{\theta})} / {\partial \bm{\theta}}$
is called the \emph{score}, and the matrix $\bm{J} := \mathsf{E} \left[ \bm{v} \bm{v}^T \right]$ is called the \emph{Fisher information matrix (FIM)} \cite{Kay1993,VanTrees2001}. Sometimes additional \emph{deterministic} \textit{a priori} information is available, which indicates that $\Theta$ is \emph{constrained} to some proper subset of the set of all $n$-dimensional complex vectors. The CRB for such cases are derived in \cite{Gorman1990,Stoica1998,Marzetta1993}. Note this \emph{deterministic} \textit{a priori} information does not result in a Bayesian setting.



Throughout the paper the complex derivative operation is defined as follows.
\begin{defn}[Complex derivative] \label{def_diff}
\begin{equation}
\frac{\partial f}{\partial z} := \frac{1}{2} \left( \frac{\partial f}{\partial \alpha} - j \frac{\partial f}{\partial \beta} \right), \quad \frac{\partial f}{\partial z^*} := \frac{1}{2} \left( \frac{\partial f}{\partial \alpha} + j \frac{\partial f}{\partial \beta} \right)\notag
\end{equation}
for any complex function $f$ and complex variable $z := \alpha + j\beta$, $\alpha, \beta \in \mathbb{R}$. 
\end{defn}
\begin{rem}
The definition is sometimes referred to as Wirtinger's calculus in literature \cite{Fritzsche2002}, and is widely adapted in the fields of complex analysis \cite{Fritzsche2002,Hormander1990} and matrix analysis \cite{Horn1994}.
\end{rem}

We define the complex score and the complex FIM based on Definition \ref{def_diff}.

\begin{defn}[Complex score and complex Fisher information matrix]
The \emph{complex score} is defined as $\bm{v}( \bm{y}; \bm{\theta} ) := {\partial \ln p( \bm{y}; \bm{\theta} )} / {\partial \bm{\theta}^*}$. The \emph{complex Fisher information matrix} is defined as $\bm{J} := \mathsf{E}[ \bm{v} \bm{v}^H ]$.
\end{defn}

We assume the regularity condition holds:
\begin{equation}
\mathsf{E}[ \bm{v} ] := \mathsf{E}\left[ \partial \ln p( \bm{y}; \bm{\theta} ) / \partial \bm{\theta}^* \right] = \bm{0}. \notag
\end{equation}
The condition is obtained by differentiating both sides of the equation $\int p( \bm{y}; \bm{\theta} ) \/ d\bm{y} = 1$ by $\bm{\theta}^*$.

\subsection{CRB for Unconstrained Complex Parameters}
The CRB for unconstrained complex parameters is a natural consequence of the above definitions.
\begin{thm}[CRB for unconstrained complex parameters]
For any unbiased estimator $\hat{\bm{\theta}}( \bm{y} )$ for the parameter $\bm{\theta}$, 
\begin{equation}
\mathsf{cov}( \hat{\bm{\theta}} ) \geq \bm{J}^{-1}. \notag
\end{equation}
The equality holds if and only if $\hat{\bm{\theta}} - \bm{\theta} = \bm{J}^{-1} \bm{v}$ in the mean square sense.
\end{thm}
\begin{proof}
We make use of the generalized Cauchy-Schwartz inequality in \cite{Rao2000}, which says
\begin{equation}
\mathsf{cov}( \bm{y}, \bm{y} ) \geq \mathsf{cov}( \bm{y}, \bm{x} ) \mathsf{cov}( \bm{x}, \bm{x} )^{-1} \mathsf{cov}( \bm{x}, \bm{y} ), \label{eq_general_CSI}
\end{equation}
and the equality holds if and only if 
\begin{equation}
\bm{y} = \mathsf{cov}(\bm{y}, \bm{x}) \mathsf{cov}(\bm{x}, \bm{x})^{-1} \bm{x} \notag
\end{equation} 
in the mean square sense. The theorem follows by substituting $\bm{x}$ and $\bm{y}$ by the score $\bm{v}$ and the estimator $\hat{\bm{\theta}}$, respectively. 
\end{proof}

\subsection{CRB for Complex Parameters Constrained by a Holomorphic Function}
First we review the definition of holomorphic functions.
\begin{defn}[Holomorphic function]
If a complex function $\bm{f}: \mathbb{C}^n \mapsto \mathbb{C}^m$ satisfies the condition
\begin{equation}
\frac{\partial \bm{f}}{\partial \bm{z}^H} = \bm{0}, \notag
\end{equation}
we call the function a \emph{holomorphic} function.
\end{defn}
\begin{rem}
A holomorphic function is also called an \emph{analytic} function in complex analysis, because a complex function is holomorphic if and only if it is analytic, although the two terms have different definitions \cite{Stein2003}.
\end{rem}

In this subsection we consider the case where the parameter space $\Theta$ is already known to be constrained by a holomorphic function $\bm{f}$. That is, 
\begin{equation}
\Theta := \{ \bm{\theta}: \bm{f}( \bm{\theta} ) = \bm{0} \}. \notag
\end{equation}
The estimator $\hat{\bm{\theta}}(\bm{y})$ is unbiased if $\mathsf{E}[ \hat{\bm{\theta}}(\bm{y}) ] = \bm{\theta}$ for all $\bm{\theta} \in \Theta$.

The CRB for complex parameters constrained by a holomorphic function is stated as follows.

\begin{thm}[CRB for complex parameters constrained by a holomorphic function] \label{thm_constrained_crb}
Let the constraint function $\bm{f}$ be a holomorphic function that maps from $\mathbb{C}^n$ to $\mathbb{C}^m$, $n \geq m$. Assume that $\partial \bm{f} / \partial \bm{\theta}^T$ has full rank. Choose $\bm{U}$ as a matrix with $(n-m)$ orthonormal columns that satisfies $\left( \partial \bm{f} / \partial \bm{\theta}^T \right) \bm{U} = \bm{0}$. Then for any unbiased estimator $\hat{\bm{\theta}}$ for the parameter $\bm{\theta}$,
\begin{equation}
\mathsf{cov}( \hat{\bm{\theta}} ) \geq \bm{U} \left( \bm{U}^H \bm{J} \bm{U} \right)^{-1} \bm{U}^H. \notag
\end{equation}
The equality holds if and only if $\hat{\bm{\theta}} - \bm{\theta} = \bm{U} \left( \bm{U}^H \bm{J} \bm{U} \right)^{-1} \bm{U}^H \bm{v}$ in the mean square sense.
\end{thm}
\begin{proof}
Since the implicit function theorem and the chain rule hold for holomorphic functions, following the derivation of \cite[Theorem I]{Marzetta1993}, we have $\mathsf{E}[ (\hat{\bm{\theta}} - \bm{\theta}) \bm{v}^* ] \bm{U} = \bm{U}$. The theorem follows by substitute $\bm{x}$ and $\bm{y}$ by $\bm{U}^H \bm{v}$ and $\hat{\bm{\theta}}$, respectively, into the generalized Cauchy-Schwartz inequality (\ref{eq_general_CSI}).
\end{proof}

\begin{rem}
This bound is valid only for parameter spaces constrained by \emph{holomorphic} functions. Interested reader shall refer to \cite{Jagannatham2004} for the general case.
\end{rem}

\begin{rem}
We call the matrix $\bm{U}$ the \emph{orthonormal complement matrix} for convenience in the following sections.
\end{rem}


\section{System Model and Problem Formulation} \label{sec_syst_model_problem_formulation}


\subsection{System Model} \label{sec_syst_model}
A discrete-time baseband model of block redundant communications through a wireless multipath channel is presented in this section \cite{Vaidyanathan2010}. 

We model the wireless multipath channel as a finite impulse response (FIR) filter with order $L$, and the noise process as complex additive white Gaussian noise (AWGN) with zero mean and unit variance. We define an $(L+1)$-dimensional complex \emph{channel vector} $\bm{h} := [ h_0 \, \ldots \, h_L ]$ to represent the channel impulse response.

At the transmitter side, the source modulation symbols are divided into blocks of $M$ symbols. Each of the block is precoded by a $P$-by-$M$ precoding matrix $\bm{F}$. In general $P$ is larger than $M$ in order to mitigate the inter-block interference; in this way redundancy is introduced. We set $P := M+L$ throughout this paper.

\begin{exmp} \label{exp_ofdm}
In a CP-OFDM system, the $P$-by-$M$ precoding matrix $\bm{F}$ is defined as
\begin{equation}
\bm{F} := \left[ \begin{array}{c} \begin{array}{cc} \bm{0}_{L\times (M-L)} & \bm{I}_{L} \end{array} \\ \bm{I}_{M} \end{array} \right] \bm{W}, \notag
\end{equation} 
where the $M$-by-$M$ matrix $\bm{W}$ is the normalized inverse discrete Fourier transform (IDFT) matrix.
\end{exmp}

We assume the receiver collects $N$ blocks for each time of channel estimation, so it is more convenient to discuss $N$ blocks of modulation symbols in a whole. Denote the source modulation symbols by an $(MN)$-dimensional complex vector $\sqrt{\gamma} \bm{s}$, the precoder output, or the \emph{channel input}, is a $(PN)$-dimensional complex vector $\bm{x} := \sqrt{\gamma} \left( \bm{I}_N \otimes \bm{F} \right) \bm{s}$. The positive constant $\gamma$ is to normalize the complex random vector $\bm{s}$ such that $\mathsf{E}[ \vert s_i \vert^2 ] = 1$ for all $i$. However, we do not assume any \emph{a priori} information about the probability distribution of source modulation symbols at the estimator.

We denote the signal that the receiver observes as a $(PN+L)$-dimensional complex vector $\bm{y}$, which, according to the channel model, is defined as
\begin{equation}
\bm{y} := \sqrt{\gamma} \, \mathcal{T}_{\bm{h}} \left( \bm{I}_N \otimes \bm{F} \right) \bm{s} + \bm{n}. \label{eq_system_model}
\end{equation}
The noise vector $\bm{n}$ is a circularly-symmetric complex Gaussian distributed random vector with mean $\bm{0}_{(PN+L)\times 1}$ and covariance matrix $\bm{I}_{PN+L}$.

\subsection{Problem Formulation} \label{sec_problem_formulation}
We show that the parameter to be estimated must satisfy a holomorphic constraint function, which is determined by
\begin{enumerate}
\item the choice of the redundant precoding matrix, and
\item the assignment of pilot symbols.
\end{enumerate}
We will derive the holomorphic constraint function explicitly.

The $(NP+L+1)$-dimensional complex parameter $\bm{\theta}$ is defined as
\begin{equation}
\bm{\theta} := \left[ \bm{h}^H \, \bm{x}^H \right]^H, \notag
\end{equation}
where the vectors $\bm{h}$ and $\bm{x}$ are the channel vector and the channel input, respectively. Since our focus is on the variance of channel estimators, the channel input $\bm{x}$ is treated as a nuisance parameter \cite{Kay1993}.

We first focus on the constraint imposed by the choice of redundant precoding matrix. We can see from the system model that the channel input $\bm{x}$ must lie in the column space of the matrix $( \bm{I} \otimes \bm{F} )$; equivalently, selecting an orthonormal matrix $\bm{U}_n$ spanning the null space of the matrix $( \bm{I} \otimes \bm{F} )$, the channel input $\bm{x}$ must satisfy
\begin{equation}
\bm{U}_n^H \bm{x} = \bm{0}, \notag
\end{equation}
which is the constraint function due to redundant precoding.

Then we turn to the constraint due to the assignment of pilot symbols. Here we adopt a general expression that by saying there are some pilot symbols in the message vector $\bm{s}$, we mean that 
\begin{equation}
\bm{A} \bm{s} = \bm{c} \label{eq_def_pilot_symbol}
\end{equation}
for some constant matrix $\bm{A}$ and constant vector $\bm{c}$.

\begin{exmp}
When the matrix $\bm{A}$ is obtained by eliminating some rows of an identity matrix, the constraint reduces to the ordinary case where some of the source modulation symbols $\bm{s}$ are assigned some pre-determined values.
\end{exmp}

By the definition of the channel input $\bm{x}$, we have
\begin{equation}
\bm{A} \left( \bm{I} \otimes \bm{F} \right)^{\dagger} \bm{x} = \bm{c}, \notag
\end{equation}
which is the constraint function due to the assignment of pilot symbols.

In conclusion, the task of a semi-blind channel estimator $\hat{\bm{h}}(\bm{y})$ is to estimate the channel coefficient vector $\bm{h}$ with deterministic \emph{a priori} information that the true parameter $\bm{\theta}$ is in a set $\Theta$ defined as
\begin{equation}
\Theta := \left\{ \bm{\theta}: \bm{\theta} \in \mathbb{C}^{NP+L+1}, \bm{f}( \bm{\theta} ) = \bm{0} \right\}, \notag
\end{equation}
where the constraint function $\bm{f}$ is
\begin{equation}
\bm{f} ( \bm{\theta} ) :=
\left[ \begin{array}{cc}
\bm{0} & \bm{0} \\
\bm{0} & \left[ \begin{array}{c} \bm{U}_n^H \\ \bm{A}\left( \bm{I} \otimes \bm{F} \right)^\dagger \end{array} \right]
\end{array} \right]
\bm{\theta} - 
\left[ \begin{array}{c} \bm{0} \\ \bm{c} \end{array} \right]. \label{eq_overall_constraint}
\end{equation}
Note that we do not put any constraint on the channel vector $\bm{h}$. Since the constraint function $\bm{f}$ is holomorphic, the simple CRB derived in Theorem \ref{thm_constrained_crb} is valid in the following derivations.

\section{CRB for Semi-Blind Channel Estimators} \label{sec_crb_for_chEst}
We first derive the CRB for unbiased estimators of the parameter vector $\bm{\theta}$. Since we are interested in the performances of semi-blind channel estimators, we then refine the CRB for the channel coefficients $\bm{h}$ from the CRB for unbiased semi-blind channel estimators of $\bm{\theta}$.

By the assumption of AWGN, we obtain the complex score as
\begin{equation}
\bm{v} := \frac{\partial \ln p( \bm{y}; \bm{\theta} ) }{\partial \theta^*} = \gamma \left[ \begin{array}{c} \mathcal{T}_{\bm{x}}^H \bm{n} \\ \mathcal{T}_{\bm{h}}^H \bm{n} \end{array} \right], \notag
\end{equation}
and the complex FIM as
\begin{equation}
\bm{J} := \mathsf{E} \left[ \bm{v} \bm{v}^H \right] = \gamma 
\left[ \begin{array}{c} \mathcal{T}_{\bm{x}}^H \\ \mathcal{T}_{\bm{h}}^H \end{array} \right] 
\left[ \begin{array}{cc} \mathcal{T}_{\bm{x}} & \mathcal{T}_{\bm{h}} \end{array} \right]. \label{eq_fim}
\end{equation}

Now we turn to the derivation of the orthonormal complement matrix in Theorem \ref{thm_constrained_crb}. Taking derivative on the overall constraint function (\ref{eq_overall_constraint}), we have
\begin{equation}
\frac{\partial \bm{f}}{\partial \bm{\theta}^T} = 
\left[ \begin{array}{cc}
\bm{0} & \bm{0} \\
\bm{0} & \left[ \begin{array}{c} \bm{U}_n^H \\ \bm{A} \left( \bm{I} \otimes \bm{F} \right)^\dagger \end{array} \right]
\end{array} \right]. \notag
\end{equation}
The orthonormal complement matrix, which we denote by $\bm{E}$, can be chosen to be a block-diagonal matrix
\begin{equation}
\bm{E} = 
\left[ \begin{array}{cc}
\bm{I} & \bm{0} \\
\bm{0} & \tilde{\bm{E}}
\end{array} \right], \label{eq_orthonormal_complement}
\end{equation}
where $\tilde{\bm{E}}$ is some orthonormal matrix that satisfies
\begin{equation}
\left[ \begin{array}{c} \bm{U}_n^H \\ \bm{A} \left( \bm{I} \otimes \bm{F} \right)^\dagger \end{array} \right] \tilde{\bm{E}} = \bm{0}. \notag
\end{equation}

According to Theorem \ref{thm_constrained_crb}, we have the following CRB for the entire parameter vector $\bm{\theta}$.
\begin{thm}[CRB for the unbiased estimators of entire parameter $\bm{\theta}$]
For any unbiased estimator $\hat{\bm{\theta}}( \bm{y} )$ for the parameter vector $\bm{\theta}$, 
\begin{equation}
\mathsf{cov}( \hat{\bm{\theta}} ) \geq \bm{E} \left( \bm{E}^H \bm{J} \bm{E} \right)^{-1} \bm{E}^H, \notag
\end{equation}
with the FIM defined as in (\ref{eq_fim}), and the orthonormal complement matrix $\bm{E}$ defined as in (\ref{eq_orthonormal_complement}).
\end{thm}

We then proceed to derive the CRB for unbiased channel estimators. The bound can be obtained as the upper-left $L$-by-$L$ sub-matrix of both sides of the CRB derived above, because if a matrix $( \bm{A} - \bm{B} )$ is nonnegative-definite, so are its principal submatrices. 

By the fact that the orthonormal complement matrix is block diagonal and its upper-left sub-matrix is the identity matrix, we can see that the upper-left $L$-by-$L$ sub-matrix of the right side of the overall CRB is exactly the same as that of the matrix
\begin{equation}
\left( \bm{E}^H \bm{J} \bm{E} \right)^{-1} = \frac{1}{\gamma}
\left( 
\left[ \begin{array}{c}
\mathcal{T}_{\bm{x}}^H \\
\tilde{\bm{E}}^H \mathcal{T}_{\bm{h}}^H
\end{array} \right]
\left[ \begin{array}{cc}
\mathcal{T}_{\bm{x}} &
\mathcal{T}_{\bm{h}} \tilde{\bm{E}}
\end{array} \right]
\right)^{-1}. \notag
\end{equation}

Therefore, the upper-left $L$-by-$L$ sub-matrix of the above matrix is the Schur complement of the upper-left $L$-by-$L$ sub-matrix of the matrix $\left( \bm{E}^H \bm{J} \bm{E} \right)$. By the definition of Schur complement, we have the CRB for any unbiased estimator for the channel coefficient vector $\bm{h}$: 
\begin{align}
&\quad \, \mathsf{cov}( \hat{\bm{h}} ) \notag \\
&\geq \frac{1}{\gamma} \left\{ \mathcal{T}_{\bm{x}}^H \left[ \bm{I} - \mathcal{T}_{\bm{h}} \tilde{\bm{E}} \left( \tilde{\bm{E}}^H \mathcal{T}_{\bm{h}}^H \mathcal{T}_{\bm{h}} \tilde{\bm{E}} \right)^\dagger \tilde{\bm{E}}^H \mathcal{T}_{\bm{h}}^H \right] \mathcal{T}_{\bm{x}} \right\}^{-1}. \label{eq_crb_h_primitive} 
\end{align}


We can further simplify the derived CRB in (\ref{eq_crb_h_primitive}). Note that for any matrix $\bm{A}$, the matrix $[ \bm{A} \left( \bm{A}^H \bm{A} \right)^\dagger \bm{A}^H ]$ is a projector to the column space of the matrix $\bm{A}$. Let $\tilde{\bm{U}}$ be a matrix with orthonormal columns that satisfies $\tilde{\bm{U}}^H ( \mathcal{T}_{\bm{h}} \tilde{\bm{E}} ) = \bm{0}$. We have the following CRB for unbiased semi-blind channel estimators.

\begin{thm}[CRB for unbiased semi-blind channel estimators]
For any unbiased semi-blind channel estimator $\hat{\bm{h}}( \bm{y} )$, 
\begin{align}
\mathsf{cov}( \hat{\bm{h}} ) \geq \frac{1}{\gamma} \left( \mathcal{T}_{\bm{x}}^H \tilde{\bm{U}} \tilde{\bm{U}}^H \mathcal{T}_{\bm{x}} \right)^{-1}, \label{eq_h_final}
\end{align}
where $\tilde{\bm{U}}$ is a matrix with orthonormal columns that spans the null space of the matrix $( \mathcal{T}( \bm{h} ) \tilde{\bm{E}} )$.
\end{thm}

The following corollary is obtained by taking traces of both sides of equation (\ref{eq_h_final}).

\begin{cor}[Variance lower bound for semi-blind channel estimators]
For any unbiased semi-blind channel estimator $\hat{\bm{h}}(\bm{y})$,
\begin{equation}
\mathsf{E} \, \Vert \hat{\bm{h}} - \bm{h} \Vert^2 \geq \frac{1}{\gamma} \, \sum_{\ell = 0}^L \sigma_{\ell}^{-2}, \label{eq_h_final_compact}
\end{equation}
where $\sigma_{\ell}$ is the $\ell$-th largest singular value of the matrix $\left( \mathcal{T}_{(\bm{I}\otimes \bm{F})\bm{s}}^H \tilde{\bm{U}} \right)$.
\end{cor}

\begin{rem}
The channel input $\bm{x}$ has been substituted by the vector $(\bm{I}\otimes \bm{F})\bm{s}$ in order to explicitly show the effect of the choice of precoding matrix to the value of CRB.
\end{rem}

\section{Conclusions} \label{sec_conclusion}
We have extended conventional CRBs, originally for \emph{real} unconstrained and constrained parameters\cite{Gorman1990,Stoica1998,VanTrees2001}, to the case of multiple \emph{complex} parameters, with simple forms. The results not only facilitate the derivation of CRB for the semi-blind channel estimation problem of interest, but also are expected to be useful for other complex CRB derivations.

Applying the simple complex CRBs, we have derived the CRB for semi-blind channel estimators in redundant block transmission systems. The derived CRB is valid for any full-rank linear redundant precoder (LRP), including the popular CP-OFDM system. The derived CRB is a lower bound on the variance of any unbiased semi-blind channel estimator, and can serve as a tractable performance metric for system design.

\bibliographystyle{IEEEtran}
\bibliography{ref}

\end{document}